\documentclass[prd,11pt,tightenlines,nofootinbib,superscriptaddress]{revtex4}
\usepackage{amsfonts,amssymb,amsthm,bbm,amsmath}
\usepackage{verbatim}
\usepackage{graphicx}
\usepackage{color}

\newcommand{\C}{{\mathbb C}}
\newcommand{\N}{{\mathbb N}}

\newcommand{\one}{\mathbbm{1}}

\newcommand{\cG}{{\mathcal G}}

\newcommand{\cL}{{\mathcal L}}
\newcommand{\cH}{{\mathcal H}}

\def\G{{\cal G}_{\Gamma}}

\newcommand{\be}{\begin{equation}}
\newcommand{\ee}{\end{equation}}
\newcommand{\beq}{\begin{eqnarray}}
\newcommand{\eeq}{\end{eqnarray}}
\newcommand{\bea}{\begin{eqnarray}}
\newcommand{\eea}{\end{eqnarray}}
\newcommand{\nn}{\nonumber}

\newcommand{\bra}{\langle}
\newcommand{\ket}{\rangle}

\newcommand{\rd}{\mathrm{d}}

\newcommand{\bpm}{\begin{pmatrix}}
\newcommand{\epm}{\end{pmatrix}}
\newcommand{\bvm}{\begin{vmatrix}}
\newcommand{\evm}{\end{vmatrix}}

\def\nn{\nonumber}

\newtheorem{theorem}{Theorem}[section]
\newtheorem{lemma}[theorem]{Lemma}

\newtheorem{definition}[theorem]{Definition}

\begin{document}

\title{Ising Model from Intertwiners}

\author{{\bf Bianca Dittrich}\email{bdittrich@perimeterinstitute.ca}}
\affiliation{ Perimeter Institute for Theoretical Physics,
Waterloo, Ontario, Canada.}
\affiliation{Department of Physics, University of Waterloo, Waterloo, Ontario, Canada}

\author{{\bf Jeff Hnybida}\email{jhnybida@perimeterinsititute.ca}}
\affiliation{ Perimeter Institute for Theoretical Physics,
Waterloo, Ontario, Canada.}
\affiliation{Department of Physics, University of Waterloo, Waterloo, Ontario, Canada}

\begin{abstract}
Spin networks appear in a number of areas, for instance in lattice gauge theories and in quantum gravity. They describe the contraction of intertwiners according to the underlying network.  We show how a certain generating function of intertwiner contractions for arbitrary networks, when restricted to a square lattice is exactly related to the high temperature expansion of the 2d Ising model partition function with constant couplings.  This implies that the intertwiner model possesses a second order phase transition, thus leading to a continuum limit with propagating degrees of freedom.
\end{abstract}

\maketitle


\section{Introduction}

Spin networks are combinatorial objects which are used in Lattice Gauge theories, Condensed matter systems, Topological Quantum Field theories, as well as models of Quantum Gravity.  They are defined simply by a directed graph $\Gamma$ decorated with spins $j_e$ i.e.\  labels for irreducible representations, on the edges and intertwiners $i_v$ on the vertices corresponding to a compact group $G$.  An $n$ valent node of $\Gamma$ is assigned an $n$-valent intertwiner which is nothing but an invariant rank $n$ tensor on the group $G$.  For a historical overview see \cite{Smolin:1997aa}.

For the choice $G=SU(2)$, the case of interest for Quantum Gravity, a coherent representation is given in terms of spinors.  In this representation $n$-valent intertwiners are labeled by $n$ spinors \cite{coh1, Freidel:2010aq,Livine:2011zz} and their contraction defines so--called coherent spin network amplitudes.  Coherent states, in general, have a special exponentiating property, which was used by Freidel and one of the authors, to construct a generating function for these intertwiner contractions \cite{Freidel:2012ji}.  

In \cite{Freidel:2012ji} two generating functions were constructed: one for the contraction of coherent intertwiners (introduced in \cite{coh1}) and one for the contraction of a new basis of intertwiners which were studied further in \cite{Freidel:2013fia}.  The first generating function was found to be expressed as the inverse square of a sum over terms in 1-1 correspondence with loops of the graph which don't share vertices or edges.  The second generating function, on the other hand, was found to be a generalization of the first in which loops were allowed to share vertices but not edges.  These generating functions were also studied in \cite{Bonzom:2012bn}.

The expression of the generating function in terms of sums over loops of the graph is reminiscent of the high temperature expansion of the Ising model.  Indeed, the 2d Ising model is defined by the configurations of spins on a 2d square lattice which can take one of two orientations.  In the high temperature expansion, the various configurations of spins can be described by loops on the dual lattice corresponding to the boundaries between domains of the two different orientations. 

By choosing the weights and graph orientations of the spin network generating function appropriately we show that one can reproduce exactly this high temperature expansion of the 2d Ising model.  The benefit of making this connection is that the 2d Ising model is exactly solvable.  Moreover, it allows us to identify a phase transition in the statistical model corresponding to the spin network generating function.

There has also been a considerable amount of work on generating functions for other spin network amplitudes, such as the Penrose evaluation \cite{penrose1971applications}.  Traditionally these amplitudes have been defined for trivalent graphs since the space of trivalent intertwiners is one dimensional, while the space of $n$-valent intertwiners, with $n>3$,  is non-trivial, but possesses bases constructed from trivalent trees.  Here the edges of the graph are labeled by the irreducible representations of SU(2) $j_e = \N/2$ and the trivalent intertwiners are unique, but possess an ordering of the edges at each vertex. 

  The amplitude associated with a trivalent spin network, referred to as the Penrose evaluation \cite{penrose1971applications}, is computed by replacing each edge with $2j_e$ strands and connecting all the strands at each node.  There are many way to do this which we will call routings, and each routing results in $N$ closed loops.\footnote{Note that if not all of the strands at a vertex can be matched then the amplitude vanishes.  This can happen if the sum of spins at the vertex is a half integer, or if the sum of two spins is less than the third.} The amplitude is then defined by
\be \label{eqn_spin_net_eval}
  A_\Gamma\big(\{j_e\}\big) \equiv  \sum_{\text{routings}} \epsilon(-2)^N
\ee  
 where $\epsilon$ is a sign which is defined such that two routings which differ by a crossing of strands have opposite sign.  Different amplitudes usually differ from this one just by a sign and normalization factor.   

Various other generating functions for the evaluations of spin networks have been developed, which can also be expressed in terms of loops.  This began with Schwinger's generating function for the $3nj$-symbols \cite{Schwinger} which Bargmann gave a more succinct presentation of in \cite{Bargmann}.  In 1975 Labarthe \cite{Labarthe:1975yf} developed a graphical method for computing the 3nj-symbol generating function for arbitrary graphs.  Then in the 1998 Westbury found a closed formula for the generating function of the chromatic evaluation on planar, trivalent graphs \cite{westbury} and shortly after by Schnetz \cite{Schnetz}.  Finally, and more recently, Garoufalidis \cite{Garoufalidis} proved the existence of the asymptotic limit while Costantino and Marche \cite{CM} solved the asymptotic evaluation and also generalized to the non-planar case, and also non-trivial holonomies.

These generating functions for amplitudes such as the Penrose evaluation (\ref{eqn_spin_net_eval}) are constructed via the variable transform $j_e \mapsto x_e$: 
\be\label{genf}
A\big(\{x_e\}\big) = \sum_{\{j_e\}}  A_\Gamma\big(\{j_e\}\big) \prod_e \frac{1}{j_e!} \, x_e^{j_e} \quad .
\ee

Let us note that (\ref{genf}) can be understood as partition function for a statistical model. Here a choice of $\{x_e\}$ amounts to a choice of statistical weights. Indeed (\ref{genf}) can be understood as a special case of so--called intertwiner models discussed in \cite{bw}.  In this work we will propose another link to a statistical model, namely the Ising model.  To do this however, we will instead employ the generating function introduced in \cite{Freidel:2012ji}.

In the next section we will shortly introduce a specific basis of $SU(2)$ intertwiners, which we will use 
to define the generating function. This section will also review the rewriting of the generating function as a sum over loop configurations $L$ obtained in \cite{Freidel:2012ji}.  This is in some sense an extension of the result of Westbury \cite{westbury} from planar trivalent graphs to arbitrary graphs, but for a slightly different evaluation, essentially differing by an overall sign.  This difference, however, has a significant effect on the generating function and is ultimately what allows us to treat higher valent nodes, in particular the square lattice.  

In hindsight, one could in fact relate Westbury's result directly the Ising model on a honeycomb lattice since it is planar and trivalent.

Section \ref{isingm} will define the partition function of the Ising model on the square lattice and give its formulations in terms of closed subgraphs $\Gamma_{\text{even}}$ with even--valent vertices. 

In the next section \ref{match} we show that for a specific choice of variables in the generating function  the loop configurations $L$ and the configurations $\Gamma_{\text{even}}$ can be matched to each other. This allows to evaluate the generating function for this specific choice of variables in terms of the partition function of the Ising model. We close with a discussion in section \ref{discuss}.

\section{Spin Network Generating Functions and the Ising Model}

In this section we review the construction of spin network generating functions as was done in \cite{Freidel:2012ji} by summing over the contraction of SU(2) intertwiners in the holomorphic representation.  In the last subsection \ref{isingm} we give a standard derivation of the high temperature expansion of the 2d Ising model.  The similarity between the two formulations both being expressed in terms of loops on the lattice should be apparent.

\subsection{Intertwiners}

First, we define a representation of SU(2) on the Bargmann-Fock space \cite{Bargmann} of holomorphic functions on spinor space $\C^2$.  This space is endowed with the Hermitian inner product
\be \label{barg_in_prod}
  \big\langle f \big| g \big\rangle = \int_{\C^2} \overline{f(z)} g(z) \rd\mu(z)
\ee
where $\rd\mu(z) = \pi^{-2} e^{-\bra z | z \ket}  \rd^{4}z$ and $\rd^{4}z$ is the Lebesgue measure on $\C^2$.
We use the notation
 $$|z\ket \equiv (\alpha, \beta )^t, \qquad |z] \equiv ( -\overline{\beta}, \overline{\alpha} )^t$$
 and $\check{z}$ to denote the conjugate spinor $|\check{z}\ket \equiv |z]$.  We use the bra-ket notation for the scalar product (\ref{barg_in_prod}) of the two states $|f\ket, |g\ket$ and a round bracket to denote $ f(z) \equiv (z|f\ket$.

The group SU(2) acts irreducibly on representations of spin $j$ given by the $2j+1$ dimensional subspaces $V^j$ of holomorphic functions homogeneous of degree $2j$.  Given a set of $n$ spins $\{j_i\}$ the space of intertwiners is defined to be
\be \label{eqn_inter_space}
  \cH_{j_1,...,j_n} \equiv \text{Inv}_{\text{SU}(2)}\left[V^{j_1} \otimes \cdots \otimes V^{j_n} \right].
\ee

In the representation space (\ref{barg_in_prod}) these are precisely the holomorphic functions of $n$ spinors $z_{1},... ,z_{n}$ which are SU(2) invariant and homogeneous of degree $2j_{i}$ in $z_{i}$ for $i=1,..,n$.  Holomorphic functions of different degree are orthogonal with respect to (\ref{barg_in_prod}) so we have
\be
  \cH_n = \bigoplus_{\{j_i\}} \cH_{j_1,...,j_n}
\ee
where $\cH_n$ is the Hilbert space of SU(2) invariant functions on $L^2(\C^{2n},\rd \mu)$.
An (overcomplete) basis of $\cH_n$ is given by monomials in the holomorphic invariants 
\be\label{C}
  \big( \{z_{i}\} \big| \{k_{ij}\} \big\rangle  \equiv \prod_{i<j}\frac{ [z_{i}|z_{j}\ket^{k_{ij}}}{k_{ij}!}
\ee
where $\{k_{ij}\}_{1\leq i,j \leq n}$ are non-negative integers and $k_{ij} = k_{ji}$.  
If $\{k\}$ satisfy the $n$ homogeneity conditions 
\be\label{kj}
\sum_{j\neq i} k_{ij} =2j_{i}.
\ee
and the sum of the spins $J \equiv \sum_i j_i = \sum_{i<j} k_{ij}$ is a positive integer then (\ref{C}) is an element of $\cH_{j_1,...,j_n}$.  The identity on ${\cal H}_{j_1,...,j_n}$ is resolved as follows
\be \label{eqn_res_id_disc}
  \one_{{\cal H}_{j_1,...,j_n}} = \sum_{\{k\}\in K_{j}} \frac{\big| \{k_{ij}\} \big\rangle \big\langle \{k_{ij}\} \big| }{\big\|\{k_{ij}\}\big\|^{2}}, \qquad \big\|\{k_{ij}\}\big\|^{2} \equiv \frac{ (J+1)!}{\prod_{i<j}k_{ij}!}.
\ee
with the set $K_j$ defined by (\ref{kj}).  For more information about this basis see \cite{Freidel:2013fia}.

The contraction of a set of intertwiners in the holomorphic representation is an operation on one or more intertwiners which arises from the natural pairing given by the scalar product (\ref{barg_in_prod}).  In general the result is another intertwiner, or when the legs of the intertwiners form a closed graph we obtain an amplitude.  If any pair of identified legs have different spins then the contraction vanishes identically.

Let $\Gamma$ be a simple\footnote{We will restrict ourselves to simple graphs here since this will be sufficient to describe a lattice.  This will allow us to label pairs of edges without a vertex label which will simplify the notation somewhat.  However, this is not a required assumption.}, closed, directed graph with edges $e$ and vertices $v$.  Assign an intertwiner $| \{k_{ee'}\} \ket \in \cH_{j_1,...,j_n}$ to each vertex where $k_{ee'}$ is defined for each pair $(e,e')$ meeting at $v$.  The contraction of the intertwiners with the connectivity given by the graph $\Gamma$ defines the amplitude
\be \label{eqn_sn_amp}
  A_\Gamma\big(\{k_{ee'}\}\big) \equiv \int \prod_{e \in \Gamma}  \rd\mu(z_{e}) \prod_{v \in \Gamma} \big( \{z^{v}_{e}\} \big| \{k_{ee'}\} \big\rangle 
\ee
where according to the edge directions we define
$$
  z^{v}_{e} \equiv \begin{cases} z_e, \quad \text{if $s_e = v$}, \\ \check{z}_{e}, \quad \text{if $t_e=v$.} \end{cases}
$$
where $s_e$, $t_e$ are the source and target vertices of the edge $e$.  

This amplitude is identically zero unless the spins of each of the contracted intertwiners are equal.  Thus by summing over all spins and contracting one can construct a generating function for the spin network amplitudes (\ref{eqn_sn_amp}) as was done in \cite{Freidel:2012ji}.  We will review this in the next section.

\subsection{Spin Network Generating functions}

As was done in \cite{Freidel:2012ji} we introduce the following generating function at a fixed vertex $v$ for the intertwiner basis (\ref{C}) which depends holomorphically on $n$ spinors $z_{e}$ for each edge and $n(n-1)/2$ complex numbers $\tau_{ee'} =-\tau_{e'e}$
\be\label{defC}
 {\cal C}_{\{\tau_{ee'}\}}\big(\{z_{e}\}\big)  \equiv e^{\sum_{e<e'} \tau_{ee'}[z_{e}|z_{e'}\ket } = \sum_{\{k\} }  \big( \{z_{e}\} \big| \{k_{ee'}\} \big\rangle \prod_{e<e'} (\tau_{ee'})^{k_{ee'}}.
\ee
Note that we have assumed an ordering of the edges $e<e'$ at the vertex.

Now given a closed, simple, directed, finite graph $\Gamma$ we attach an intertwiner generating function (\ref{defC}) to each vertex and integrate over the spinors with the measure (\ref{barg_in_prod}).  This will give us a generating generating function for the spin network amplitudes for the contraction of the basis states (\ref{C}).

To be more precise, we choose the convention that if two vertices are connected by an edge $e$ then the vertex with the outgoing direction is assigned $z_{e}$ while the vertex with incoming direction is assigned $\check{z}_e$ which is defined by $|\check{z}_{e}\ket = |z_e]$ as in (\ref{eqn_sn_amp}).  To each pair of edges we assign a complex number $\tau_{ee'}=-\tau_{e'e}$ and we define
\be\label{defG}
\G\big(\{\tau_{ee'}\}\big) \equiv  \int \prod_{e\in \Gamma} \rd\mu(z_{e}) \prod_{v\in \Gamma} e^{\sum_{(e,e') \supset v} \tau_{ee'}[z^{v}_{e}|z^{v}_{e'}\ket } 
 = \sum_{\{k\}} A_\Gamma\big(\{k_{ee'}\}\big) \prod_{v} \prod_{(e,e') \supset v} (\tau_{ee'})^{k_{ee'}} 
\ee
where the integral is over one spinor per edge of $\Gamma$ and the sum is over pairs of edges $(e,e')$ meeting at $v$.  Note that an ordering of the edges at each vertex is assumed since each pair $(e,e')$ is associated to $\tau_{ee'}$.  In what follows, for planar graphs, we will use either a cyclic or acyclic ordering at each vertex, meaning there is a reference edge and edges are ordered either clockwise or counterclockwise from that reference.

Since $\G(\tau)$ is expressed as a Gaussian integral we can perform these integrations which results in the determinant of a matrix with $\tau_{ee'}$ for elements.  Furthermore, as was shown in \cite{Freidel:2012ji} this determinant can be evaluated as a sum of terms, which are in one to one correspondence with certain loops of $\Gamma$ which we now define:

\begin{definition}
A loop of $\Gamma$ is a sequence of edges $l= (e_{1}, \cdots, e_{n})$ with a cyclic ordering such that $ t_{e_{i}}= s_{e_{i+1}}$ and $ t_{e_{n}}= s_{e_{1}}$.  We will identify loops which just differ by a permutation of edges leading to the same cyclic ordering of the edges, as well as loops just differing in the two possible directions.  A simple loop of $\Gamma$ is a loop of $\Gamma$ in which each edge enters at most once.  We say two
simple loops are disjoint if they have no edges in common.
\end{definition}


Note that a simple loop can intersect itself at vertices and disjoint simple loops can also intersect at vertices.  This will be important for graphs with valence greater than four since there are various ways in which paths can cross.  (See Figure \ref{fig_simpleloops}).  Finally, we state a key result of \cite{Freidel:2012ji} in which the evaluation of  (\ref{defG})is expressed as a sum over all collections of disjoint simple loops.  For a proof see \cite{Freidel:2012ji}.

\begin{theorem} \label{thm_gen}
The generating function (\ref{defG}) has the evaluation
\be  \label{eqn_gen_loops}
  \G(\tau) = \frac1{\left(1 + \sum_{L} A_{L}(\tau)\right)^{2}}
\ee
where the sum is over all collections of disjoint simple loops of $\Gamma$.  For each collection $L = \{\ell_1,...,\ell_k\}$ we define $A_{L}(\tau) = A_{\ell_{1}}(\tau)\cdots A_{\ell_{k}}(\tau)$ where for each simple loop $\ell_i = \{e_{1}, \cdots ,e_{n}\}$ we define the quantity
\be \label{eqn_A_loop}
A_{\ell}(\tau) \equiv -(-1)^{|e|} \tau_{e_{1}^{-1} e_2} \tau_{e_{2}^{-1} e_3} \cdots \tau_{e_{n}^{-1} e_1}
\ee
where $|e|$ is the number of edges of $\ell$ whose orientation agrees with the chosen orientation of  $\Gamma$.
\end{theorem}

This representation of the spin network generating function is reminiscent of the high temperature expansion of the 2d Ising model partition function.  We will show in the next section that if we take $\Gamma$ to be a square lattice with a particular choice of weights and orientations, we can find an exact relationship.

\begin{figure} 
  \centering
    \includegraphics[width=0.7\textwidth]{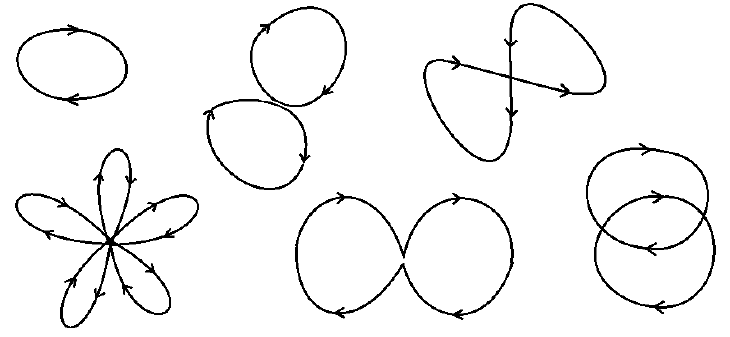}
    \caption{Some examples of paths on a graph which are collections of disjoint simple loops.  Notice that the middle three diagrams each have an intersection of four edges at one vertex, but but they follow different paths.}  \label{fig_simpleloops}
\end{figure}

\subsection{The 2d Ising Model}\label{isingm}

Following \cite{mussardo} we review how the 2d Ising model can also be formulated in terms of simple loops on a lattice.  The 2d Ising model on a square 2d lattice describes the possible configurations of spins placed on the lattice sites which can take one of two possible orientations.  The intuition of R. Peierls \cite{peierls} was that the possible states of this model are given by all possible loops on the dual lattice, which represent the boundary between domains of aligned spins.  The energy associated with the creation of such a domain is given by
$$
  \Delta E = 2JL
$$ 
where $J$ is a coupling constant and $L$ is the number links in the boundary of the domain.  The partition function of the Ising model on a square lattice $\cL_N$ of size $N \times N$at zero magnetic field with one coupling constant $J$ is
\be \label{eqn_ising_part}
  Z_N(v) = \sum_{\{\sigma\}} \exp \left(  \beta J \sum_{(i,j)} \sigma_i \sigma_j \right)
\ee
where for each vertex $i$ the spins are $\sigma_i = \pm 1$, and the sum in the exponent is over nearest neighbors.  Using the identity
$$
  \exp(x \sigma_i \sigma_j) = \cosh (x) ( 1- \sigma_i \sigma_j \tanh (x) )
$$
we get
\be
  Z_N(v) = \cosh^N( J) \sum_{\{\sigma\}} \prod_{(i,j)} (1 - \tanh( \beta J) \sigma_i \sigma_j)
\ee 
Expanding the product and defining $v \equiv \tanh \beta J$ 
\be \label{eqn_ising_loops}
  Z_N(v) = 2^N (1-v^2)^{-N} \left(1+\sum_{P \geq 4} g_P v^P \right)
\ee
where $g_P$ is the number of closed subgraphs $\Gamma^P_{\text{even}}$ of the lattice $\cL$ having a total of $P$ links and with an even number of edges adjacent to each vertex.  In what follows we define $\Gamma_{\text{even}}$ to be the set of such closed, even--valent subgraphs having an arbitrary number of links $P\geq 4$.  Here a closed subgraph can have disconnected components which share neither edges nor vertices.  For more details see \cite{mussardo}.

\begin{figure} 
  \centering
    \includegraphics[width=0.45\textwidth]{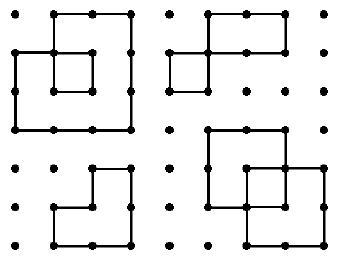}
    \caption{The closed even valent subgraphs of the square lattice correspond to domain boundaries in the 2d Ising model.}  \label{fig_polygon}
\end{figure}

\section{Matching of generating function and Ising model} \label{match}

So far we have reviewed the construction of the spin network generating function (\ref{eqn_gen_loops}) and the 2d Ising model partition function (\ref{eqn_ising_loops}), both in terms of sums of loops on a graph.  We now show that for a particular choice of the parameters $\tau_{ee'}$ and orientation in the spin network generating function we can produce 2d Ising model partition function exactly.

To see this, first note that the sum over collections of disjoint simple loops in Theorem \ref{thm_gen} on a square lattice contains all of these configurations of closed subgraphs $\Gamma_{\text{even}}$ in (\ref{eqn_ising_loops}), but also more due to the three possible ways in which two paths can cross at a four--valent vertex\footnote{See the middle three diagrams of Figure \ref{fig_simpleloops}.}.  Another difference is that there are signs in (\ref{eqn_gen_loops}) due to the edge orientation and the vertex ordering.  However, for a particular choice of edge orientation and vertex ordering of the square lattice, and a homogeneous choice of weights $\tau_{ee'} = i \sigma_{ee'} v$, with $\sigma_{ee'}$ being an antisymmetric function, the two sums are equal, as we now show.

\begin{theorem} \label{cor_gen}
Let $\cL$ be the square lattice with edge orientation and vertex ordering as in Figure \ref{fig_lattice}.  Let the vertex weights in (\ref{defG}) be given homogeneously by $\tau_{ee'} = iv$ for $e<e'$ and $\tau_{ee'} = -iv$ for $e>e'$.  Then the spin network generating function (\ref{defG}) takes the form
\be \label{eqn_G_L}
{\cal G}_{\cL}(iv) = \left(1+\sum_{P} g_P v^P \right)^{-2}
\ee
where the sum is over all even--valent, closed subgraphs of $\cL$ as in (\ref{eqn_ising_loops}).
\end{theorem}

\begin{figure} 
  \centering
    \includegraphics[width=0.7\textwidth]{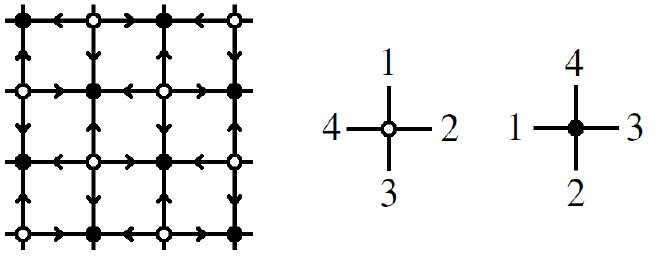}
    \caption{The edge orientation and vertex ordering of a square lattice for which the terms in (\ref{eqn_gen_loops}) all have a positive sign as shown in Theorem \ref{cor_gen}.}  \label{fig_lattice}
\end{figure}

We will prove this theorem by a series of lemmas.  The first step is to control the signs in (\ref{eqn_gen_loops}) which is accomplished by the specific edge direction and vertex ordering in Figure \ref{fig_lattice}.  We say that a vertex $v$ in a loop disagrees with the vertex ordering, if the loop traverses first the edge $e$ and then the edge $e'$ adjacent to $v$ and $e' < e$.  Furthermore a loop without crossing is a loop which may have self intersections (i.e. four edges of the loop meet at one vertex), however the edges are traversed without leading to crossing edge pairs. 

\begin{lemma} \label{lemma_orient}
Let $\cL$ be the lattice in Figure \ref{fig_lattice} with the indicated edge orientation and vertex ordering.  Then 
\begin{enumerate}
\item the number of edges in a loop which agrees with the orientation of $\cL$ is equal to half the number of edges in the loop
\item the number of vertices in a loop without self--crossing, which disagrees with the vertex ordering is odd.
\end{enumerate}
\end{lemma}
\begin{proof}
For the first part, it is easy to see that the edges of every loop in $\cL$ alternates orientation and every loop has an even number of edges so the number of edges that agrees with the orientation is equal to half the number of edges in the loop.



For the second part, we will use induction on the number of plaquettes in the lattice.  To this end we will build up the lattice from the left most lower corner.  One can add squares so that the boundary on the right forms a staircase to reach an infinite lattice in the limit. A finite size lattice can be built row by row. We thus have two cases to consider: adding a square which starts a new row and adding a square to an existing row as is illustrated in Figure \ref{fig_lattice2}.
Notice furthermore that the ordering along a vertex is reversed if the loop is reversed, hence we need just to consider one specific loop orientation. Furthermore exchanging all black vertices with white ones and vice versa we also exchange all orientation induced signs, hence we again just need to consider one choice for the partitioning of the vertices into black and white.

One can check that the loop on a single square has an odd number of vertices which disagrees with the vertex ordering.  Assume that we have a square lattice for which every loop has an odd number of vertices which disagree.  
Consider adding a single square starting a new  row, as in the left panel of  Figure \ref{fig_lattice2}.
By the hypothesis all of the loops which contain $e_1$ have an odd number of vertices which disagree.  Traversing $e_1$ in any direction gives one vertex which disagrees.  On the other hand traversing the three edges in the new square clockwise gives three vertices which disagree (or one vertex that disagrees in the counter-clockwise direction).  Furthermore the new square might lead to a loop with a non--crossing self intersection at the black vertex $v_1$, shared by $e_1$. Here one can also check that for a counter--clockwise orientation of the loop a deformation of the loop to include the new square leads to four additional vertices that disagree.  
Hence all loops of the lattice with the new square also have an odd number of vertices which disagree.  

 Similarly, for adding a square to an existing row, one can check that traversing $e_2$, $e_3$ (or both) contributes the same parity as traversing the new square. Again one can also check that  loops with non--crossing self intersections at the black or white vertex of $e_3$, which include the new square, have an odd number of vertices disagreeing with the ordering. 
 
 Hence by induction the loops in a square lattice of any size will always have an odd number of vertices that disagrees with the ordering.
\end{proof}

\begin{figure} 
  \centering
    \includegraphics[width=0.55\textwidth]{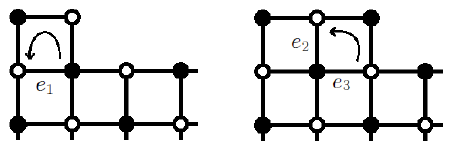}
    \caption{Adding one square to the lattice: Starting a new row and adding to an existing row.  Assuming all the loops in the existing lattice have an odd number of vertices which disagree with the edge ordering, then the loops containing the new square also have an odd number.}  \label{fig_lattice2}
\end{figure}

We have now to discuss the situation that at a given vertex either one loop self--intersects, or two loops touch or even cross each other. A priori all these cases are allowed to appear in the sum for the generating function (\ref{eqn_gen_loops}). This leads to three terms for such a vertex, as there are three possibilities for how two paths meet or cross at a four--valent vertex (see the middle three diagrams of Figure \ref{fig_simpleloops}). In the partition function of the Ising model (\ref{eqn_ising_loops}) only one term for  such a vertex appears. Hence we have to show that always two terms cancel each other, and that the surviving term does not lead to a loop with crossing.


\begin{lemma} \label{lemma_STU}
Consider the lattice $\cL$ and let 
\be
  \tau_{ee'} =\sigma_{ee'}\tau_{e}\tau_{e'}  
\ee
where $\sigma_{ee'} = 1$ if $e<e'$ and $\sigma_{ee'} = -1$ if $e>e'$ according to the vertex ordering.  

 Then the sum over collections of disjoint simple loops in (\ref{eqn_gen_loops}) is reduced to only collections without crossings.  Furthermore, there is a one--to--one matching between these terms and configurations $\Gamma_{\text{even}}$ of closed, even--valent subgraphs.  
\end{lemma}

\begin{proof}

Suppose we have a configuration $A_{L}(\tau)$ of disjoint simple loops, for which all four edges $e_1,\ldots,e_4$ adjacent to a vertex $v$ are shared by either one or two loops. The way the loop or the two loops traverses the four edges, leads to a partition of the  four edges into two pairs of consecutive edges in the loop(s).  There are three such possible pairings. The crossing case  $(1-3,2-4)$ and the two non--crossing cases $(1-2,3-4)$ and $(2-3,4-1)$. (Here the ordering of the edges inside a pair does not matter.)

Hence  there are also two other configurations, which include the same set of edges as $A_{L}(\tau)$, but differ by a certain rearrangement of the edges into loops, so that the other two pairings are obtained. This gives three configurations, which we will name $A_U$ for the crossing case,  $A_S$ for $(1-2,3-4)$  and $A_T$ for $(2-3,4-1)$.

To be concrete consider a black vertex, for white vertices one just has to invert the edges $e_1,\ldots e_4$ everywhere. Note that under a change of orientation of a simple loop we have $A_{\ell}=A^{-1}_{\ell}$ due to the anti--symmetry of the $\tau_{ee'}$ and the definition (\ref{eqn_A_loop}). Furthermore we can choose w.l.o.g.\ the initial vertex in any given loop. Hence we can assume that in the configuration $A_U$ we have a loop $\ell_{U1}$ of the form $\ell_{U1}=(e_3^{-1} P P' e_1)$ where $P$ and $P'$ stand for paths with the source vertex $s(P)$ given by $t(e_3^{-1})$ and the target vertex of $P'$ being $t(P')=s(e_1)$.

We now consider three possibilities for the end point of $P$.
\begin{itemize}\parskip-2mm
\item[(a)]  We have that the target vertex $t(P)=s(e_2)=s(P')$ with $P'=(e_2 e_4^{-1} p')$.
 \item[(b)] We have $t(P)=s(e_4)=s(P')$ with $P'=(e_4 e_2^{-1} p')$. 
\item[(c)] We have $t(P)=s(e_1)$.
In this case $P'$ is empty and there is a second loop $\ell_{U2}$ contributing to $A_U$ whose orientation and starting point we can choose such that $\ell_{U2}=( e_4^{-1} p'e_2)$ with $s(p')=t(e_4^{-1})$ and $t(p')=s(e_2)$. (The two loops intersect also elsewhere for a planar lattice.)
\end{itemize}

Let us  define the corresponding configurations $A_S$ and $A_T$ for the different cases. 
\begin{itemize}\parskip-2mm
\item[(a)] $A_S$ agrees with $A_U$ in all simple loops except for $\ell_{U1}$ which is replaced by $\ell_{S}=(e_2^{-1}P^{-1}e_3e_4^{-1}p'e_1)$. Likewise we replace for $A_T$ the loop  $\ell_{U1}$ by two loops $\ell_{T}=(e_4^{-1}p'e_1)(e_3^{-1}Pe_2)$.
\item[(b)] For $A_S$ we replace  $\ell_{U1}$ by a pair of loops $\ell_S=(e_{2}^{-1}p'e_1)(e_4^{-1}P^{-1}e_3)$ and for $A_T$ by a loop $\ell_T=(e_4^{-1}P^{-1}e_3e_2^{-1}p'e_1)$.
\item[(c)]  For $A_S$ we replace $\ell_{U1}\ell_{U2}$ by a loop $\ell_S=(e_2^{-1} (p')^{-1} e_4e_3^{-1}Pe_1)$ and for $A_T$ by a loop $\ell_T=(e_4^{-1}p'e_2e_3^{-1} Pe_1)$.
\end{itemize}

\begin{figure} 
  \centering
    \includegraphics[width=1\textwidth]{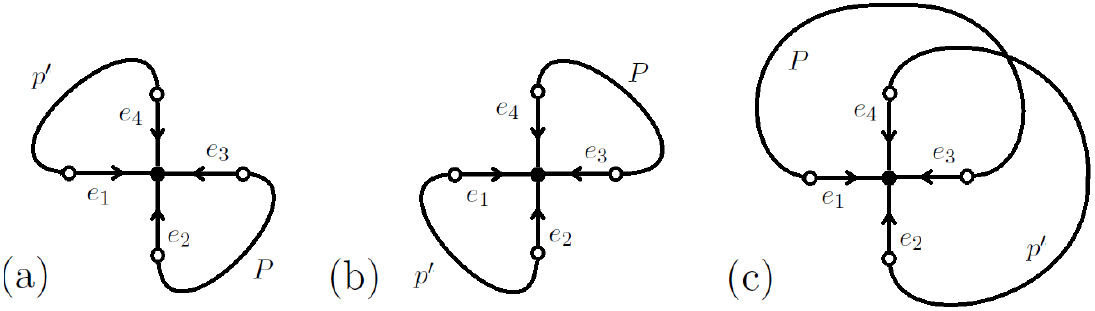}
    \caption{The three possible intersections at a 4-valent vertex.  For each intersection (a), (b), and (c) there are three possible configurations of simple loops $S$, $T$, and $U$.  The paths $P$ and $p'$ are arbitrary.}  \label{fig_abc}
\end{figure}

We have now to compare the corresponding amplitudes as defined in (\ref{eqn_A_loop}). To this end denote by 
\be
A_{{\cal P}} \,= \, (-)^{|{\cal P}|} \prod_{{ \text{bulk}} \,\, v} \tau_v 
\ee
 the contribution from an open path ${\cal P}$, where $|{\cal P}|$ is the number of edges disagreeing with the orientation of the path and $\tau_v$ stands for $\tau_{ee'}$ with $(e,e')$ a pair of edges in ${\cal P}$ adjacent to $v$ and ordered according to the orientation of $|{\cal P}|$. 
 
 Note that under a reversal of the orientation of ${\cal P}$ we have 
 \be
 A_{\cal P} \,=\, (-)A_{{\cal P}^{-1}} \quad .
 \ee
 The reason for this is that the change in sign due to the orientation of edges is given by $(-1)^{\sharp {\cal P}}$ where $\sharp {\cal P}$ is the number of edges in ${\cal P}$. Furthermore the change in sign due to the orientation of the vertices and the antisymmetry of the $\tau_{ee'}$ is given by $(-1)^{\sharp {\cal P}+1}$.
 
 We can now consider all three cases:\\
With $A_{\text{rest}}$ denoting the contribution of all other simple loops in $A_U$  we obtain for the case (a)
 \bea
 A_U&=& A_{\text{rest}}   (-) A_{(e_3^{-1}Pe_2)} \tau_{e_2 e_4^{-1}} A_{(e_4^{-1}p'e_1)} \tau_{e_1e_3^{-1}} \nn\\
 A_S&=&A_{\text{rest}}   (-) A_{(e_3^{-1}Pe_2)^{-1}} \tau_{e_3e_4^{-1}} A_{(e_4^{-1}p'e_1)} \tau_{e_1e_2^{-1}}\,=\,-A_U \nn\\
 A_T &=&A_{\text{rest}}   (-) A_{(e_3^{-1}Pe_2)} \tau_{e_2 e_3^{-1}} (-) A_{(e_4^{-1}p'e_1)} \tau_{e_1e_4^{-1}}\,=\,-A_U \quad .
 \eea
 Here we used the special form of the weights $\tau_{ee'}=\sigma_{ee'}\tau_{e}\tau_{e'}$ to reach $A_S=A_T=-A_U$.
 
 Likewise we also obtain for the other two cases (b) and (c) that $A_S=A_T=-A_U$.

  Thus for cases (a) and (b) we can  cancel in the sum $\sum_L A_L(\tau)$ the term with a crossing $A_U$ such that we remain with the contribution of two simple loops, i.e. for (a) we cancel $A_U$ with $A_S$ and for (b) we cancel $A_U$ with $A_T$. 
  
  In the case (c) we have to cancel $A_U$ with either $A_S$ or $A_T$ and we remain with a loop $\ell_S$ or $\ell_T$ with self--intersection (but non--crossing) at the vertex $v$ under consideration. 
  
  However, in the case of a planar lattice  the two loops $\ell_{U1}$ and $\ell_{U2}$ need to cross at least one other time at one or more other vertices $v',v'',\ldots$. Going to the next vertex, for instance $v'$, we can now resolve this crossing so that the loop is split into two loops. The self--intersection of $\ell_S$ or $\ell_T$ at $v$ then turns into two different loops sharing two vertices. 
  
  Doing this with all vertices we remain with loops which do not self--intersect. Different loops may share vertices. Counting all such configurations would still lead to an over--counting compared to the number of configurations of closed graphs $\Gamma_{\text{even}}$, as can be seen by an example of two loops sharing two vertices\footnote{See the lower right diagram in Figure \ref{fig_polygon}.}, for which there are two (if the loops are not crossing) possibilities involving the same set of edges. But in fact the proof shows that resolving all intersections leads always to just one configuration that remains in the end.  This leads to a matching of (left--over) loops configurations with configurations of closed, even--valent subgraphs $\Gamma_{\text{even}}$ for the Ising model.  
	
Remark: The fact that from the three possible terms $A_S,A_T,A_U$ two terms cancel out generalizes to arbitrary lattices. However to specify the crossing term $A_U$ one needs a planar vertex. Furthermore for i.e. six--valent vertices, three paths might meet at one vertex, in which case one has more terms to consider.	
\end{proof}


Now Theorem \ref{cor_gen} follows from lemmas \ref{lemma_orient} and \ref{lemma_STU}.  Indeed, from Lemma \ref{lemma_STU} the sum in (\ref{eqn_gen_loops}) is reduced to a sum of terms in one--to--one correspondence with the subgraphs $\Gamma_{\text{even}}$ and each term in the sum is a collection of disjoint simple loops having no crossings.  Suppose such a subgraph has $P$ edges then by lemma \ref{lemma_orient} the quantity (\ref{eqn_A_loop}) will have a sign $(-1)^{P/2}$ which is canceled by the factors of $i$ in the weight.


This gives us the following relation between the spin network generating function and the 2d Ising model partition function.
\be \label{eqn_G_Z}
  \cG_{\cL_N}(iv) = \frac{2^N}{ (1-v^2)^{N} Z_{N}(v)^2}
\ee
where $Z_N(v)$ is the partition function (\ref{eqn_ising_part}) of the 2d Ising model.  In particular, this shows that in the limit $N \rightarrow \infty$ the spin network generating function $\cG_{\cL_N}(iv)$ possesses a second order phase transition at $$v = \sqrt{2}-1$$

Indeed, it is known that the 2d Ising model undergoes a second order phase transition for a particular temperature, namely when $v = \sqrt{2}-1$.  The free energy of $Z_N(v)$ is defined by
\be
  F(T) = -kT \log Z_N(v)
\ee
and is exactly solvable for $N \rightarrow \infty$.  At the critical temperature the logarithm in $F(T)$ becomes singular and since
\be
  \log \cG_{\cL_N}(iv) = N\log 2 - N \log(1-v^2) -2 \log Z_N(v)
\ee
it follows that the logarithm of $\cG_{\cL_N}(iv)$ is also singular at this point.  Thus we have shown that the spin network generating function $\cG_{\cL_N}(iv)$ will undergo a second order phase transition at the critical value $v = \sqrt{2}-1$.

\section{Conclusion} \label{discuss}

We have shown that spin network generating functions can encode partition functions of statistical models, in this case, the Ising model. The solvability of the 2D Ising model allowed us to obtain an explicit expression for the spin network generating function -- for a specific choice of arguments.   

This shows that the rewriting of the spin network generating function as a sum over loops obtained in \cite{Freidel:2012ji}, see also \cite{Garoufalidis}, can help to understand properties of this generating function. As mentioned such generating functions can be understood as statistical models themselves, where the variables in the generating function encode the weights of the models. As such generating functions are related to intertwiner models considered in \cite{bw} for which choices of weights were identified which lead to topological models, that is a continuum limit without propagating degrees of freedom. Here we specified weights leading to a second order phase transition, i.e the continuum limit gives a theory with propagating degrees of freedom.

Here we started with an intertwiner model, which admits a geometrical interpretation (similar to spin foam models) of the underlying spin variables, due to the triangle inequalities satisfied by the spins meeting at a vertex, see also \cite{bw}. We have shown that for a specific choice of weights, i.e. choice of dynamics,  we obtain a continuum limit with propagating degrees of freedom. This model can therefore serve as a toy example, describing 2D geometries, for many conceptual questions about  coarse graining spin foam models, see also \cite{refining}. It will be in particular interesting to study the meaning of the macroscopic order parameters and correlation functions emerging from the Ising model description, in the original microscopic model. 

Apart from the spin network generating functions defining 2D partition functions, they also appear describing boundary states of 3D quantum gravity, more precisely for the Ponzano Regge model \cite{PR}. Although the Ponzano Regge model is a topological model, its partition function can be (holographically) dual to a boundary theory with propagating degrees of freedom \cite{BCL,BD}. To uncover this dual field theory it is necessary to take the continuum limit of the boundary discretization. The results in this work show that the Ponzano Regge model admits boundary states that lead to a non--trivial boundary field theory.

We hope that the techniques presented here can be extended to map to other known statistical models, also involving irregular lattices. Also lattices with different topology might be treatable if one introduces certain defects or special vertices. Even if these models are not exactly solvable, such maps would provide many tools to understand the properties of spin network generating functions. Such generating functions appear \cite{jefftoappear} if spin foam or spin net models are coarse grained \cite{nets,foams}.  For an interesting and different approach see \cite{Livine:2013gna}. The understanding of the coarse graining flow for spin foams, which aim at a description of quantum gravity, is a key problem for quantum gravity research \cite{alexreview}. This can be also very much understood as a problem of statistical physics, with models defined on regular lattices \cite{nets,bw}, as we are using here.  We hope therefore that the technique developed here will eventually help to understand the phase diagram for spin foams.

 We would like to note that during the final stages of this work, we became aware that similar results are being developed by Valentin Bonzom, Francesco
Costantino and Etera  Livine.  These results have appeared (after this work's appearance on arXiv) in \cite{BCL}.

\acknowledgments

We would like to thank Laurent Freidel, Etera Livine, and Miklos L\r{a}ngvik for helpful discussions and we also thank Laurent Freidel for encouragement to write this paper.
Research at Perimeter Institute is supported
by the Government of Canada through Industry Canada and by the Province of Ontario
through the Ministry of Research and Innovation.

\appendix

\end{document}